\documentclass[a4paper, 10pt]{article}
\usepackage[T1]{fontenc}

\usepackage{amsmath,bbm,amssymb} 
\usepackage{amsthm}
\usepackage[toc,page]{appendix}
\usepackage{url}
\usepackage{authblk}

\newtheorem{theorem}{Theorem}[section]

\newtheorem{defi}[theorem]{Definition}

\newtheorem{proposition}[theorem]{Proposition}

\newtheorem{Rem}[theorem]{Remark}

\title{About the decomposition of pricing formulas under stochastic volatility models.}
\author{Raúl Merino\footnote{Universitat de Barcelona, Facultat de Matemàtiques, Gran Via 585, 08007 Barcelona, Spain} \footnote{VidaCaixa S.A., Investment Control Department, Juan Gris, 20-26, 08014 Barcelona, Spain.\\ 
E-MAIL: raul.merino85@gmail.com} \hspace{5mm} Josep Vives\footnote{Universitat de Barcelona, Facultat de Matemàtiques, Gran Via 585, 08007 Barcelona, Spain.\\ 
E-MAIL: josep.vives@ub.edu }}

\date{}
\begin{document}

\maketitle
\begin{abstract}
We obtain a decomposition of the call option price for a very general stochastic volatility diffusion model extending the decomposition obtained by E. Alòs in \cite{Alos} for the Heston model. We realize that a new term arises when the stock price does not follow an exponential model. The techniques used are non anticipative. In particular, we see also that equivalent results can be obtained using Functional Itô Calculus. Using the same generalizing ideas we also extend to non exponential models the alternative call option price decompostion formula obtained in \cite{Alos06} and \cite{ALV} written in terms of the Malliavin derivative of the volatility process. Finally, we give a general expression for the derivative of the implied volatility under both, the anticipative and the non anticipative case. 

\end{abstract}

\section{Introduction}
\hspace{0.4cm} Stochastic Volatility models are a natural extension of the Black-Scholes model in order to manage the skew and the smile observed in real data. It is well known that in these models the average of future volatilities is a relevant quantity. Unfortunately adding a stochastic volatility structure, makes pricing and calibration more complicated, due to the fact that closed formulas not always exist and even when this formulas exist, in general, don't allow a fast calibration of the parameters.

During the last years different developments for finding approximations to the closed-form option pricing formulas have been published. Malliavin techniques are naturally used to solve this problem in \cite{Alos06} and \cite{ALV} as the average future volatility is an anticipative quantity. Otherwise, a non anticipative method to obtain an approximation of the pricing formula is developed for the Heston model in \cite{Alos}. The method is based on the use of the adapted projection of the average future volatility and obtain a decomposition of the call option price in terms of it.   

In the present paper we generalize the results of \cite{Alos} to general stochastic volatility diffusion models. Similarly, following the same kind of ideas, we extend the expansion based on Malliavin calcululs obtained in \cite{Alos06} and \cite{ALV}. 

The main ideas developed in this paper are the following:  

\begin{itemize}
	\item A generic call option price decomposition is found without the need to specify the volatility structure. 
	\item A new term emerges when the stock option prices does not follow an exponential model, as for example in the SABR case. 
	\item The Feynman-Kac formula is a key element in the decomposition. It allows to express the new terms that emerges under the new framework (i.e. stochastic volatility) as corrections of the Black-Scholes formula.
	\item The decompostion found using Functional Itô calculus turns out to be the same as the decomposition obtained by our techniques. 
	\item We give a general expression of the derivative of the implied volatility, both for the non anticipative and the anticipative cases.  
\end{itemize}

\section{Notation.}
\hspace{0.4cm} Let $S=\{S(t), t\in [0,T]\}$ be a strictly positive price process under a market chosen risk neutral probability that follows the model:  

\begin{eqnarray}\label{general_model}
 dS(t)= \mu(t,S(t))dt + \theta(t, S(t), \sigma(t)) \left(\rho dW(t)  + \sqrt{1-\rho^{2}}dB(t) \right)
\end{eqnarray}
where $W$ and $B$ are independent Brownian motions, $\rho \in (-1,1)$, $\mu$: $[0,T] \times {\mathbb R}_+ \rightarrow \mathbb{R}$, $\theta$: $[0,T] \times {\mathbb R}_{+}^{2} \rightarrow {\mathbb R}_{+}$ and $\sigma(t)$ is a positive square-integrable process adapted to the filtration of $W$. We assume on $\mu$ and $\sigma$ sufficient conditions to ensure the existence and uniqueness of the solution of \eqref{general_model}. Notice that we don't assume any concrete volatility structure. Thus, our decompositions can be adapted to many different models. In particular we cover the following models: 

\begin{itemize}
\item
Black-Scholes model: $\mu(t,S(t)):=rS(t)$, $\theta(t, S(t), \sigma(t)):=\sigma S(t)$, $\rho=0$, $r>0$ and $\sigma>0$.

\item
CEV model: $\mu(t,S(t)):=rS(t)$, $\theta(t, S(t), \sigma(t)):=\sigma S(t)^{\beta}$ with $\beta\in (0,1]$, $\rho=0$, $r>0$ and $\sigma>0$.

\item
Heston model: $\mu(t,S(t)):=rS(t)$, $\theta(t, S(t), \sigma(t)):=\sigma(t) S(t)$, $r>0$, $\sigma>0$ and 

\begin{eqnarray}\label{heston}
d\sigma^{2}(t) &=& k(\theta-\sigma^{2}(t)) dt + \nu \sqrt{\sigma^{2}(t)}dW(t),
\end{eqnarray}
where $k$, $\theta$ and $\nu$ are positive constants satisfiyng the Feller condition $2k\theta>\nu^{2}$.

\item
SABR model: $\mu(t,S(t)):=rS(t)$, $\theta(t, S(t), \sigma(t)):=\sigma(t) S(t)^{\beta}$ with $\beta\in (0,1]$, $r>0$, $\sigma>0$ and 
\begin{eqnarray}\label{SABR}
d\sigma(t)=\alpha \sigma(t)dW(t)  
\end{eqnarray}
with $\alpha>0.$
\end{itemize}

For existence and unicity of the solution in the Heston case see for example \cite{G}, Section 2.2. For the CEV and SABR models see \cite{COW} and the references therein.  

The following notation will be used in all the paper: 

\begin{itemize}
	
	\item We will denote by $BS(t,S,\sigma)$ the price of a plain vanilla european call option under the classical Black-Scholes model with constant 
	volatility $\sigma$, current stock price $S$, time to maturity $\tau=T-t$, strike price $K$ and interest rate $r$. In this case, 
\begin{eqnarray}
\nonumber BS\left(t,S, \sigma\right)= S \Phi(d_{+}) - K e^{-r\tau} \Phi(d_{-}),
\end{eqnarray}
where $\Phi(\cdot)$ denotes the cumulative probability function of the standard normal law and 
\begin{eqnarray}
\nonumber d_{\pm} = \frac{\ln(S/K) + (r \pm \frac{\sigma^{2}}{2})\tau}{\sigma\sqrt{\tau}}.
\end{eqnarray}

	\item We use in all the paper the notation $\mathbb{E}_{t}[\cdot]:=\mathbb{E}[\cdot|\mathcal{F}_{t}]$, where $\left\{\mathcal{F}_{t}, t\geq 0\right\}$ is the natural filtration of $S.$ 

	\item In our setting, the call option price is given by 
	
	$$V(t)=e^{-r\tau}{\mathbb E}_t [(S(T)-K)^+].$$
		
	\item  Recall that from the Feynman-Kac formula, the operator 
	\begin{eqnarray}{\label{FK}}
	 \mathcal{L}_{\theta}:= {\partial}_t + \frac{1}{2}\theta(t,S(t),\sigma(t))^{2} {\partial^{2}_{S}} 
							+ \mu(t,S(t)) {\partial}_{S}-r
	\end{eqnarray}
	satisfies ${\mathcal L}_{\theta}BS(t,S(t),\theta(t,S(t),\sigma(t)))=0$.
	
	\item
	We will also use the following definitions for $y\geq 0$: 
	
	$$G(t, S(t), y):= S^{2}(t) \partial^{2}_{S} BS (t, S(t), y),$$

	$$H(t, S(t), y):= S(t)\partial_{S}G(t, S(t), y),$$ 

	$$K(t, S(t), y):= S^{2}(t)\partial^{2}_{S}G(t, S(t), y)$$
	and 

	$$L(t,S(t), y):=\frac{\theta(t,S(t),y)}{S(t)}.$$

\end{itemize}

\section{A decomposition formula using Itô Calculus.}

\hspace{0.4cm} In this section, following the ideas in \cite{Alos}, we extend the decomposition formula to a generic stochastic volatility diffusion process. We note that the new formula can be extended without the need to specify the underlying volatility process, obtaining a more flexible decomposition formula. When the stock price does not follow an exponential process a new term emerges. The formula proved in 
\cite{Alos} is a particular case.

It is well known that if the stochastic volatility process is independent from the price process, the pricing formula of a plain vanilla European call is given by 

$$V(t)={\mathbb E}_t [BS(t,S(t),{\bar \sigma}(t))]$$  
where ${\bar \sigma}^2(t)$ is the so called average future variance and it is defined by 

\begin{eqnarray}
\nonumber {\bar\sigma}^2(t):=\frac{1}{T-t} \int^{T}_{t}\sigma^{2}(s)ds.
\end{eqnarray}
Naturally, ${\bar\sigma}(t)$ is called the average future volatility. See \cite{FPS}, pag. 51.   

The idea used in \cite{Alos} consists in consider the adapted projection of the average future variance 

\begin{eqnarray}
\nonumber v^2(t):={\mathbb E}_t({\bar\sigma}^2(t))=\frac{1}{T-t} \int^{T}_{t}{\mathbb E}_t[\sigma^{2}(s)]ds.
\end{eqnarray}
and obtain a decomposition of $V(t)$ in terms of $v(t).$ This idea switches an anticipative problem related with the anticipative process ${\bar\sigma}(t)$ into a 
non-anticipative one with the adapted process $v(t).$ We apply this technique to our generic stochastic differential equation \eqref{general_model}.

\allowdisplaybreaks
\begin{theorem} \label{Teo-Ito}
(Decomposition formula) For all $t\in [0,T)$ we have
\begin{eqnarray}
\nonumber	V(t)&=& BS(t,S(t),v(t))\\ \nonumber
&+&\frac{1}{2}\mathbb{E}_{t}\left[\int^{T}_{t} e^{-r(u-t)} G(u, S(u), v(u)) \left(L^{2}(u,S(u), \sigma(u))- \sigma^{2}(u)\right)du\right]\\ \nonumber
&+&\frac{1}{8}\mathbb{E}_{t}\left[\int^{T}_{t} e^{-r(u-t)} K(u, S(u), v(u)) d\left[M,M\right](u)\right]\\ \nonumber
&+&\frac{\rho}{2}\mathbb{E}_{t}\left[\int^{T}_{t} e^{-r(u-t)} L(u, S(u),\sigma(u)) H(u, S(u), v(u)) d\left[W,M\right](u)\right]
\end{eqnarray}
where $M(t):=\int^{T}_{0} \mathbb{E}_t\left[\sigma^{2}(s)\right]ds=\int^{t}_{0}\sigma^{2}(s)ds+(T-t)v(t)^2.$

\end{theorem}

\begin{proof}
Notice that $e^{-rT} BS(T,S(T),v(T))=e^{-rT} V(T)$. As $e^{-rt}V(t)$ is a martingale we can write
\begin{eqnarray}
\nonumber e^{-rt}V(t) = \mathbb{E}_{t}\left(e^{-rT}V(T)\right)= \mathbb{E}_{t}\left(e^{-rT}BS(T, S(T), v(T))\right).
\end{eqnarray}
Our idea is to apply the Itô formula to the process $e^{-rt}BS(t,S(t), v(t)).$ 

As a consequence of the fact that the derivatives of $BS$ are not bounded we have to use an aproximation to the identity argument changing $BS(t,S,\sigma)$ by
\begin{eqnarray}
\nonumber BS_{n}(t,S,\sigma):=BS(t,S,\sigma)\psi_{n}(S) 
\end{eqnarray}
where $\psi_{n}(S)=\phi(\frac{1}{n}S)$ for some $\phi \in \mathcal{C}^{2}_{b}$ such that $\phi(S)=1$ for all $|S|<1$ and $\phi(S)=0$ for all $|S|>2$, and $v(t)$ by 
$v^{\varepsilon}(t)=\sqrt{\frac{1}{T-t} \left(\delta + \int^{T}_{t} \mathbb{E}\left[\sigma^{2}(s) ds\right]\right)}$, where  $\varepsilon>0$, and apply finally 
the dominated convergence theorem. For simplicity we skip this mollifier argument in all the paper. 

\

So, applying the Itô formula, using the fact that 

\begin{equation}\label{rel-vega}
\partial_{\sigma}BS(t,S,\sigma)=S^2\sigma\tau \partial^{2}_{S}BS(t,S,\sigma)
\end{equation}
and the Feynman-Kac operator \eqref{FK}, we deduce
\begin{eqnarray}	
\nonumber && e^{-rT}BS(T,S(T),v(T))  - e^{-rt}BS(t,S(t),v(t))=\\ \nonumber
&=& \int^{T}_{t} e^{-ru} \mathcal{L}_{vS}BS(u,S(u),v(u))du \\ \nonumber
&+&\int^{T}_{t} e^{-ru} \partial_{S}BS(u,S(u),v(u)) \theta(u,S(u),\sigma(u)) \left(\rho dW(u)  + \sqrt{1-\rho^{2}}dB(u) \right)\\ \nonumber
&+& \frac{1}{2}\int^{T}_{t} e^{-ru} S^{2}(u)\partial^{2}_{S}BS(u,S(u),v(u))dM(u)\\ \nonumber 
&+& \frac{1}{2}\int^{T}_{t} e^{-ru} S^{2}(u)\partial^{2}_{S}BS(u,S(u),v(u))\left[L^{2}(u,S(u), \sigma(u))du - \sigma^{2}(u)du\right]\\ \nonumber 
&+&\frac{1}{8}\int^{T}_{t} e^{-ru} \left(S^{2}(u)\partial^{2}_{S}\left(S^{2}(u)\partial^{2}_{S}BS(u,S(u),v(u))\right)\right) d\left[M,M\right](u)\\ \nonumber
&+&\frac{\rho}{2}\int^{T}_{t} e^{-ru} \theta(u, S(u), \sigma(u))\left(\partial_{S}\left(S^{2}(u)\partial^{2}_{S}BS(u,S(u),v(u))\right)\right) d\left[W,M\right](u).	
\end{eqnarray}

Taking conditional expectation and multiplying by $e^{rt}$, we have:
\begin{eqnarray}	
\nonumber &&\mathbb{E}_{t}[e^{-r(T-t)}BS(T,S(T),v(T))]=  BS(t,S(t),v(t))\\ \nonumber
&+& \frac{1}{2} \mathbb{E}_{t} \left[\int^{T}_{t} e^{-r(u-t)} S^{2}(u) \partial^{2}_{S} BS (u, S(u), v(u)) \left(L^{2}(u,S(u), \sigma(u))du - \sigma^{2}(u)\right)  du\right]\\ \nonumber
&+&\frac{1}{8}\mathbb{E}_{t}\left[\int^{T}_{t} e^{-r(u-t)} \left(S^{2}(u)\partial^{2}_{S}\left(S^{2}(u)\partial^{2}_{S}BS(u,S(u),v(u))\right)\right) d\left[M,M\right](u)\right]\\ \nonumber
&+&\frac{\rho}{2}\mathbb{E}_{t}\left[\int^{T}_{t} e^{-r(u-t)} \theta(u, S(u), \sigma(u)) \partial_{S}\left(S^{2}(u)\partial^{2}_{S}BS(u,S(u),v(u))\right) d\left[W,M\right](u)\right].	
\end{eqnarray}

\end{proof}

\begin{Rem}
In \cite{Alos}, the following operators are defined for $X(t)=\log S(t)$
	
\begin{itemize}

	\item $\tilde{G}(t,X(t),\sigma(t)):=\left(\partial^{2}_{x} -\partial_{x}\right)BS(t, X(t), \sigma(t)).$
	
	\item $\tilde{H}(t,X(t),\sigma(t)):=\left(\partial^{3}_{x} -\partial^{2}_{x}\right)BS(t, X(t), \sigma(t)).$
	
	\item $\tilde{K}(t,X(t),\sigma(t)):=\left(\partial^{4}_{x}-2\partial^{3}_{x} + \partial^{2}_{x}\right)BS(t,X(t),\sigma(t)).$

\end{itemize}

We observe that 
\begin{itemize}
	\item $\tilde{G}(t,X(t),\sigma(t))= G(t,S(t),\sigma(t)).$

	\item $\tilde{K}(t,X(t),\sigma(t))= K(t,S(t),\sigma(t)).$

	\item $\tilde{H}(t,X(t),\sigma(t))=H(t,S(t),\sigma(t))$.

\end{itemize}
\end{Rem}

\begin{Rem}
We have extended the decomposition formula in \cite{Alos} to the generic SDE \eqref{general_model}. When we apply the Itô calculus, we realize that Feynman-Kac formula absorbs some of the terms that emerges. Finally, we ended up with three new terms to adjust the price. It is important to note that this technique works for any payoff or any diffusion model satisfying Feynman-Kac formula.
\end{Rem}

\begin{Rem}
Note that when $\theta(t,S(T), \sigma(t)) = \sigma(t) S(t)$ (i.e. the stock price follows an exponential process) then 
\begin{eqnarray}
\nonumber	V(t)&=& BS(t,S(t),v(t))\\ \nonumber
&+&\frac{1}{8}\mathbb{E}_{t}\left[\int^{T}_{t} e^{-r(u-t)} K(u, S(u), v(u)) d\left[M,M\right](u)\right]\\ \nonumber
&+&\frac{\rho}{2}\mathbb{E}_{t}\left[\int^{T}_{t} e^{-r(u-t)} \sigma(u) H(u, S(u), v(u)) d\left[W,M\right](u)\right],
\end{eqnarray}
and the term
\begin{eqnarray}
\nonumber \frac{1}{2} \mathbb{E}_{t} \left[\int^{T}_{t} e^{-r(u-t)} S^{2}(u) \partial^{2}_{S} BS (u, S(u), v(u)) \left(L^{2}(u,S(u), \sigma(u)) - \sigma^{2}(u)\right) du\right] 
\end{eqnarray}
vanishes. 

Indeed, we will show that due to the use of Feynman-Kac formula this is happening.
\begin{eqnarray}
\nonumber &&\text{Movement of the asset }  + \text{ Movement of the volatility } \\ \nonumber
&=&\frac{1}{2}\theta(u,S_u)^{2} \partial^{2}_{S}BS(u,S(u),v(u))du + \partial_\sigma BS(u,S(u),v(u)) dv(u)\\ \nonumber
&=&\frac{1}{2}\sigma^{2}(u)S^{2}(u) \partial^{2}_{S}BS(u,S(u),v(u))du \\ \nonumber
&+&\frac{1}{2}S^{2}(u)\partial^{2}_{S}BS(u,S(u),v(u)) (dM + v^{2}du - \sigma^{2}du) \\ \nonumber
&=&\frac{1}{2}S^{2}(u)\partial^{2}_{S}BS(u,S(u),v(u)) (dM + v^{2}du)
\end{eqnarray}
where $\frac{1}{2}S^{2}(u)\partial^{2}_{S}BS(u,S(u),v(u)) v^{2}$ is used into the Feynman-Kac formula and $$\mathbb{E}_{t}\left[\int^{T}_{t}\frac{1}{2}S^{2}\partial^{2}_{S}BS(u,S(u),v(u))dM\right]=0.$$
\end{Rem}

\section{Basic elements of Functional Itô calculus.}
\hspace{0.4cm} In this section we give the insights of the Functional Itô calculus developed in \cite{Cont1, Cont2, Cont3, Dupire}.\\

\par

Let $X:[0,T] \times\Omega \longmapsto \mathbb{R}$ be an Itô process, i.e. a continuous semimartingale defined on a filtered probability space $(\Omega, \mathcal{F},(\mathcal{F}_{t})_{t\in[0,T]}, \mathbb{P})$ which admits the stochastic integral representation
\begin{eqnarray}
X(t) = x_0+\int^{t}_{0} \mu(u) du + \int^{t}_{0} \sigma(u)dW(u)
\end{eqnarray}
where W is a Brownian motion and $\mu(t)$ and $\sigma(t)$ are continuous processes respectively in $L^1 (\Omega\times [0,T])$ and $L^2(\Omega\times [0,T]).$

We define $D([0,T],\mathbb{R})$ the space of cadlag functions. Given a path $x\in D([0,T],\mathbb{R})$, we will denote as $x_{t}$ its restriction to $[0,t]$. For $h\geq 0$, the \textbf{\textsl{horizontal extension}} $x_{t,h}$ is defined as
\begin{eqnarray}
x_{t,h}(u)=x_{t}(u)=x(u), \hspace{2mm} u\in[0,t[ ;\hspace{4mm} x_{t,h}(u)=x(t), \hspace{2mm} u\in(t,t+h]
\end{eqnarray}
and the \textbf{\textsl{vertical extension}} as 
\begin{eqnarray}
x^{h}_{t}(u)&=&x_{t}(u)=x(u), \hspace{2mm} u\in[0,t[;\\ \nonumber
x^{h}_{t}(t)&=&x(t) + h,  \text{ i.e. } x^{h}_{t}(u)=x(u) + h\mathbbm{1}_{\left\{t=u\right\}}.
\end{eqnarray}
A process $Y: [0,T] \times \mathbb{R}\rightarrow \mathbb{R},$ progressively measurable with respect the natural filtration of $X$, may be represented as
\begin{eqnarray}
\nonumber Y(t)=F\left(t, \left\{X(s), 0 \leq s\leq t \right\}\right)=F(t, X_{t})
\end{eqnarray}
for a certain measurable functional $F:[0,t]\times D([0,t], \mathbb{R})\rightarrow \mathbb{R}$.
Let ${\mathbb F}^{\infty}$ be the space of locally lipschitz functionals with respect the norm of the supremum on $D([0,t+h],{\mathbb R})$, that is, 
it exists a constant $C>0$ such that for any compact $K$ and for any $x\in D([0,t],K)$ and $y\in D([0,t+h],K)$ we have 

$$|F(t, x_t)-F(t+h, y_{t+h})|\leq C ||x_{t,h}-y_{t+h}||_{\infty}.$$ 

Under this framework, we have the next definitions of derivative:
\begin{defi} (Horizontal Derivative) The horizontal derivative of a functional $F\in {{\mathbb F}^{\infty}}$ at $t$ is defined  as
\begin{eqnarray}
\mathcal{D}_{t}F(t, x_{t})=\lim_{h \rightarrow 0^{+}} \frac{F(t+h, x_{t,h}) - F(t, x_{t})}{h}.
\end{eqnarray}
\end{defi}

\begin{defi}(Vertical derivative) 
The vertical derivative of a functional $F\in {{\mathbb F}^{\infty}}$ at $t$ is defined  as
\begin{eqnarray}
\nabla_{x} F(t, x_{t}) = \lim_{h \rightarrow 0^{+}} \frac{F(t, x^{h}_{t}) - F(t, x_{t})}{h}.
\end{eqnarray}
Of course we can consider iterated derivatives as $\nabla_{xx}.$
\end{defi}

We also have the following Itô formula that works for non-anticipative functionals:

\begin{theorem}(Functional Itô Formula) For any non-anticipative functional $F\in {{\mathbb F}^{\infty}}$ and any 
$t\in[0,T]$ we have

\begin{eqnarray}
\nonumber F(t, X_{t})-F(0, X_{0})&=&\int^{t}_{0} \mathcal{D}_{u} F(u, X_{u}) du + \int^{t}_{0} \nabla_{x}F(u, X_{u}) dX(u)\\ \nonumber
&+&\frac{1}{2} \int^{t}_{0} \nabla_{xx}F(u, X_{u}) d\left[ X,X\right](u),
\end{eqnarray}
provided $\mathcal{D}_{t}F$, $\nabla_{x}F$ and $\nabla_{xx}F$ belong to ${\mathbb F}^{\infty}.$
\end{theorem}

\begin{proof}
See \cite{Cont2, Cont3}.
\end{proof}

\section{A general decompostion using Functional Itô Calculus.}

\hspace{0.4cm} In this section, we apply the technique of functional Itô calculus to the problem of finding a decomposition for the call option price. The decomposition problem is an anticipative path-dependent problem, using a smart choice of the volatility process into the Black-Scholes formula we can convert it into a non anticipative one. It is natural to wonder if the functional Itô calculus brings some new insides into the problem.

We consider the functional 
\begin{eqnarray}
\nonumber F(t, S(t), \sigma^{2}_{t})= e^{-rt}BS(t,S(t), f(t, \sigma^{2}_{t})) 
\end{eqnarray}
where $\sigma^2$ is the path-dependent process and $f\in {\mathbb F}^{\infty}$ is a non-anticipative functional.

Under this framework, we calculate the derivatives using the functional Itô calculus respect the variance and we write them in terms of the classical Black-Scholes derivatives. We must realize that, for simplicity, the new derivatives are calculated respect the variance instead of the volatility of the process.

\begin{Rem}
If $\partial$ denotes the classical derivative, we have: 

\begin{itemize}

\item Alternative Vega :
\begin{eqnarray}
 \nonumber \nabla_{\sigma^{2}}F = e^{-rt}\partial_{f} BS(t, S(t),  f(t, \sigma^{2}_{t})) \hspace{2mm}\nabla_{\sigma^{2}} f(t, \sigma^{2}_{t}).
\end{eqnarray}

\item Alternative Vanna:
\begin{eqnarray}
\nonumber \nabla_{\sigma^{2},S}F = \partial_{f,S} BS(t, S(t),  f(t, \sigma^{2}_{t})) \hspace{2mm} \nabla_{\sigma^{2}} f(t, \sigma^{2}_{t}).
\end{eqnarray}

\item Alternative Vomma:
\begin{eqnarray*}
\nabla_{\sigma^{2},\sigma^{2}}F &=& e^{-rt}\partial_{f,f}BS(t, S(t),  f(t, \sigma^{2}_{t}))\left(\nabla_{\sigma^{2}} f(t, \sigma^{2}_{t})\right)^{2}\\
&-& e^{-rt}\partial_{f}BS(t, S(t),  f(t, \sigma^{2}_{t}))\nabla^{2}_{\sigma^{2}}f(t, \sigma^{2}_{t}).
\end{eqnarray*}

\item Alternative Theta:
\begin{eqnarray*}
\nonumber \mathcal{D}_{t} F &=&-re^{-rt}BS(t, S(t),  f(t, \sigma^{2}_{t}))\\
&+& e^{-rt} \partial_t BS(t, S(t),  f(t, \sigma^{2}_{t}))\\
&+& e^{-rt}\partial_{f} BS(t, S(t),  f(t, \sigma^{2}_{t}))\mathcal{D}_{t} f(t, \sigma^{2}_{t}) .
\end{eqnarray*}

\end{itemize}
\end{Rem}

\allowdisplaybreaks
\begin{theorem} \label{Teo-functional}
(Decomposition formula) For all $t\in [0,T)$, $S(t)$ and $f(t, \sigma^{2}_{t})>0$ we have
\begin{eqnarray}
\nonumber &&V(t)= BS(t, S(t), f(u, \sigma^{2}_{t}))\\ \nonumber
&+&\mathbb{E}_{t} \left[\int^{T}_{t}e^{-r(u-t)}f(u, \sigma^{2}_{u}) \tau G(u, S(u),  f(u, \sigma^{2}_{u}))\mathcal{D}_{u} f(u, \sigma^{2}_{u})du\right] \\ \nonumber
&+&\frac{1}{2}\mathbb{E}_{t} \left[\int^{T}_{t}e^{-r(u-t)}G(u, S(u),  f(u, \sigma^{2}_{u}))\left(L^{2}(u,S(u), \sigma_{u}) - f^{2}(u, \sigma^{2}_{u})\right)du \right] \\ \nonumber
&+& \frac{1}{2}\mathbb{E}_{t} \left[\int^{T}_{t}e^{-r(u-t)}f(u, \sigma^{2}_{u})^{2}\tau^{2}K(u, S(u),  f(u, \sigma^{2}_{u})) d\left[ f(u, \sigma^{2}_{u}),f(u, \sigma^{2}_{u})\right]\right]\\ \nonumber
&+& \rho \mathbb{E}_{t} \left[\int^{T}_{t}e^{-r(u-t)}L^{2}(u,S(u), \sigma(u))H(u, S(u), f(u, \sigma^{2}_{u})) f(u, \sigma^{2}_{u}) \tau d\left[W(u),  f(t, \sigma^{2}_{u})\right]\right].
\end{eqnarray}
\end{theorem}

\begin{proof}
Notice that $F(T,X(T), \sigma^{2}_T)=e^{-rT} BS(T,S(T),f(T, \sigma^{2}_{T}))=e^{-rT} V_{T}$. As $e^{-rt}V(t)$ is a martingale we can write
\begin{eqnarray*}
\nonumber e^{-rt}V(t) &=& \mathbb{E}_{t}\left(e^{-rT}V(T)\right)\\
&=& \mathbb{E}_{t}\left(e^{-rT}BS(T, S(T), f(T, \sigma^{2}_{T}))\right)\\
&=&\mathbb{E}_{t}\left(F(T,S(T), \sigma^{2}_{T})\right).
\end{eqnarray*}
Our idea is to apply aproximation to the identity argument as in Theorem \ref{Teo-Ito} and then use the functional Itô formula to 
$$F(t,S(t), \sigma^{2}_{t})=e^{-rt}BS(t,S, f(t, \sigma^{2}_{t})).$$

We deduce that 
\begin{eqnarray}
\nonumber &&F(T,S(T), \sigma^{2}_{T})- F(t,S(t),  \sigma^{2}_{t})\\ \nonumber
&=& \int^{T}_{t} \mathcal{D}_{u}  F(u, S(u), \sigma^{2}_{u})du + \int^{T}_{t} \nabla_{S} F(u, S(u), \sigma^{2}_{u}) dS(u)\\ \nonumber
&+& \int^{T}_{t} \nabla_{\sigma^{2}} F(u, S(u), \sigma^{2}_{u})du + \frac{1}{2}\int^{T}_{t}\nabla^{2}_{S} F(t, S(t), \sigma^{2}_{t}) d[S,S](u)\\ \nonumber
&+& \frac{1}{2}\int^{T}_{t}\nabla^{2}_{\sigma^{2}} F(u, S(u), \sigma^{2}_{u}) d\left[\sigma^{2}_{u}, \sigma^{2}_{u}\right] + \frac{1}{2}\int^{T}_{t}\mathcal{D}^{2}_{u}  F(u, S(u), \sigma^{2}_{u})du \\ \nonumber
&+& \int^{T}_{t}\nabla^{2}_{S,\sigma^{2}} F(u, S(u), \sigma^{2}_{u}) d\left[S(u), \sigma^{2}_{u}\right] \\ \nonumber
&+& \int^{T}_{t}\partial_{f} \left(\nabla_{S} F(u, S(u), \sigma^{2}_{u})\right) d\left[S(u),  f(u, \sigma^{2}_{u})\right]\\ \nonumber
&+&\int^{T}_{t}\partial_{f}\left(\nabla_{\sigma^{2}} F(u, S(u), \sigma^{2}_{u})\right) d\left[\sigma^{2},  f(u, \sigma^{2}_{u})\right].
\end{eqnarray}

Note that:
\begin{itemize}
	\item As $S(t)$ is not path-dependent, we have that $\nabla_{S}(\cdot)= \partial_{S}(\cdot)$.
	
	\item As $u>t$ and $f$ is a non-anticipative functional, then $\nabla_{\sigma^{2}(u)} f(t, \sigma^{2}_{t})=0$.
\end{itemize}
So, we have 
\begin{eqnarray}
\nonumber &&F(T,S(T), \sigma^{2}_{T})- F(t,S(t),  \sigma^{2}_{t})\\ \nonumber
&=& \int^{T}_{t} \mathcal{D}_{u}  F(u, S(u), \sigma^{2}_{u})du + \int^{T}_{t} \partial_{S} F(u, S(u), \sigma^{2}_{u}) dS(u)\\ \nonumber
&+& \frac{1}{2}\int^{T}_{t}\partial^{2}_{S} F(u, S(u), \sigma^{2}_{u}) d\left[S, S\right](u) + \frac{1}{2}\int^{T}_{t}\mathcal{D}^{2}_{u}  F(u, S(u), \sigma^{2}_{u})du \\ \nonumber
&+&  \int^{T}_{t} \partial^{2}_{f,S} F(u, S(u), \sigma^{2}_{u}) d\left[S(u),  f(u, \sigma^{2}_{u})\right].
\end{eqnarray}
We deduce that
\begin{eqnarray}
\nonumber &&F(T,S(T), \sigma^{2}_{T})- F(t,S(t),  \sigma^{2}_{t})\\ \nonumber
&=&  \int^{T}_{t}\mathcal{L}_{f(u, \sigma^{2}_{u})}BS du  +\int^{T}_{t}e^{-ru}\partial_{f} BS(u, S(u),  f(u, \sigma^{2}_{u}))\mathcal{D}_{u} f(u, \sigma^{2}_{u}) du\\ \nonumber
&+& \frac{1}{2}\int^{T}_{t}e^{-ru}\partial^{2}_{S} BS(u, S(u),  f(u, \sigma^{2}_{u}))\left(\theta^{2}(u,S(u), \sigma(u)) - S^{2}f^{2}(u, \sigma^{2}_{u})\right)du \\ \nonumber
&+& \int^{T}_{t} \partial_{S} BS(u, S(u), f(u, \sigma^{2}_{u})) \theta(u,S(u), \sigma(u)) \left(\rho dW(u)  + \sqrt{1-\rho^{2}}dB(u) \right)\\ \nonumber
&+& \frac{1}{2}\int^{T}_{t}e^{-ru}\partial^{2}_{f} BS(u, S(u),  f(u, \sigma^{2}_{u}))d \left[f(u, \sigma^{2}_{u}) , f(u, \sigma^{2}_{u})\right]\\ \nonumber
&+& \rho \int^{T}_{t}e^{-ru}\partial^{2}_{f,S} BS(u, S(u), f(u, \sigma^{2}_{u})) \theta(u,S(u), \sigma(u)) d\left[W(u),  f(u, \sigma^{2}_{u})\right].
\end{eqnarray}
Taking now conditional expectations, using \eqref{rel-vega} and multiplying by $e^{rt}$ we obtain that
\begin{eqnarray}
\nonumber &&e^{-r(T-t)}\mathbb{E}_{t} \left[F(T,S(T), \sigma^{2}_{T})\right]= BS(t, S(t), f(u, \sigma^{2}_{t}))\\ \nonumber
&+&\mathbb{E}_{t} \left[\int^{T}_{t}e^{-ru}f(u, \sigma^{2}_{u}) (T-t) S^{2}\partial^{2}_{S} BS(u, S(u),  f(u, \sigma^{2}_{u}))\mathcal{D}_{u} f(u, \sigma^{2}_{u})\right]du \\ \nonumber
&+&\frac{1}{2}\mathbb{E}_{t} \left[\int^{T}_{t}e^{-ru}G(u, S(u),  f(u, \sigma^{2}_{u}))\left(L^2(u,S(u), \sigma(u)) - f^{2}(u, \sigma^{2}_{u})\right)du \right] \\ \nonumber
&+& \frac{1}{2}\mathbb{E}_{t} \left[\int^{T}_{t}e^{-ru}f(u, \sigma^{2}_{u})^{2}\tau^{2}K(u, S(u),  f(u, \sigma^{2}_{u})) d\left[f(u, \sigma^{2}_{u}) , f(u, \sigma^{2}_{u})\right]\right]\\ \nonumber
&+& \rho \mathbb{E}_{t} \left[\int^{T}_{t}e^{-ru}H(u, S(u),  f(u, \sigma^{2}_{u}))f(u, \sigma^{2}_{u}) \tau\theta(u,S(u), \sigma(u)) d\left[W(u),  f(u, \sigma^{2}_{u})\right]\right].
\end{eqnarray}
\end{proof}

\begin{Rem}
Note that functional Itô formula proved in \cite{Cont2} holds for semimartingales, but in \cite{Cont3} is also proved for Dirichlet process. In both cases, the hypothesis hold by definition of $f$ and differentiability of the derivatives of Black-Scholes function when $\tau,S,\sigma>0$. Therefore, this technique can be applied to these models.
\end{Rem}

\begin{Rem}
Note that Theorem \ref{Teo-functional} coincides with Theorem \ref{Teo-Ito} when we choose the volatility function as $f(t, \sigma^{2}_{t})=v(t)$. We found an equivalence of the ideas developed by \cite{Cont1, Cont2, Cont3, Dupire} and \cite{Alos} in the decomposition problem. Both formulas come from very different places, the ideas under \cite{Cont1, Cont2, Cont3, Dupire} are based on an extension to functionals of the work \cite{Follmer}, while the main idea of  \cite{Alos} is to change a process by his expectation. Realise that standard Itô calculus also can be applied to Dirichlet process (for more information see \cite{Follmer}).
\end{Rem}

\begin{Rem}
Realize that Theorem \ref{Teo-functional} holds for any non-anticipative $f(t, \sigma^{2}_{t})$. It is no trivial to find a different non anticipative process $f(t, \sigma^{2}_{t})$ different from the one chosen in \cite{Alos}.
\end{Rem}

\section{Basic elements of Malliavin Calculus.}
\hspace{0.4cm} In the next section, we present a brief introduction to the basic facts of Malliavin calculus. For more information, see \cite{N}.\\

\par

Let us consider a Brownian motion $W=\left\{W(t), t\in \left[0,T\right]\right\}$ defined on a complete probability space $\left(\Omega, \mathcal{F}, \mathbb{P}\right)$. Set $H=L^{2}(\left[0,T\right])$, and denote by $W(h)$ the Wiener integral of a function $h\in H$. Let $\mathcal{S}$ be the set of random variables of the form $F=f(W(h_{1}, \ldots, W(h_{n}))$, where $n\geq1$, $f \in \mathcal{C}^{\infty}_{b}$, and $h_{1}, \ldots , h_{n} \in H$. Given a random variable $F$ of this form, we define its derivative as the stochastic process $\left\{D^{W}_{t}F, t \in[0,T]\right\}$ given by
\begin{eqnarray}
D^{W}_{t}F = \sum^{n}_{i=1} \partial_{x_{i}}f(W(h_{1}), \ldots, W(h_{n}))h_{i}(x), \hspace{0.4cm} t \in \left[0,T\right].
\end{eqnarray}
The operator $D^{W}$ and the iterated operators $D^{W,n}$ are closable and unbounded from $L^{2}(\Omega)$ into $L^{2}([0,T]^{n} \times \Omega)$, for all $n \geq 1$. We denote the closure of $\mathcal{S}$ with respect to the norm
\begin{eqnarray}
\left\|F\right\|^{2}_{n,2} := \left\|F\right\|^{2}_{L^{2}(\Omega)} + \sum^{n}_{k=1} \left\|D^{W,k} F\right\|^{2}_{L^{2}([0,T]^{k} \times \Omega)} .
\end{eqnarray}
We denote by $\delta^{W}$ the adjoint of the derivative operator $D^{W}$. Notice that $\delta^{W}$ is an extension of the Itô integral in the sense that the set $L^{2}_{a}([0,T] \times \Omega)$ of square integrable and adapted processes is included in $Dom\delta$ and the operator $\delta$ restricted to $L^{2}_{a}([0,T] \times \Omega)$ coincides with the Itô stochastic integral. We use the notation $\delta(u)=\int^{T}_{0}u(t)dW(t)$. We recall that $\mathbb{L}^{n,2}_{W}:= L^{2}\left([0,T]; \mathbb{D}^{n,2}_{W}\right)$ is contained in the domain of $\delta$ for all $n\geq1$.

We will use the next Itô formula for anticipative processes.

\begin{proposition}\label{MallaivinIto}
Let us consider the processes $X(t) = x(0) + \int^{t}_{0} u(s) dW(s) + \int^{t}_{0} v(s)ds$, where $u, v \in \textit{L}^{2}_{a}([0, T] \times \Omega)$. Furthermore consider also a process $Y(t) = \int^{T}_{t} \theta(s) ds$, for some $\theta \in \mathbb{L}^{1,2}$.Let $F: \mathbb{R}^{3} \rightarrow \mathbb{R}$ a twice continuously differentiable function such that there exists a positive constant $C$ such that, for all $t \in [0,T]$, $F$ and its derivatives evaluated in $(t, X(t), Y(t))$ are bounded by $C$. Then it follows
\begin{eqnarray}
\nonumber F(t,X(t), Y(t)) &=& F(0, X(0), Y(0)) + \int^{t}_{0} \partial_{s} F (s, X(s), Y(s)) ds\\ \nonumber
&+& \int^{t}_{0} \partial_{x}F(s, X(s), Y(s)) dX(s)\\
&+& \int^{t}_{0} \partial_{y}F(s, X(s), Y(s)) dY(s) \\ \nonumber 
&+& \int^{t}_{0} \partial^{2}_{x,y} F(s, X(s), Y(s)) (D^{-} Y)(s) u(s) ds \\ \nonumber
&+& \frac{1}{2} \int^{t}_{0} \partial^{2}_{x} F(s, X(s), Y(s)) u^{2}(s) ds,
\end{eqnarray} 
where $(D^{-}Y)(s) := \int^{T}_{s} D^{W}_{s} Y(r) dr$. 
\end{proposition}
\begin{proof}
See \cite{Alos06}.
\end{proof}

The next proposition is useful when we want to calculate the Malliavin derivative.
\begin{proposition}
Let $\sigma$ and $b$ be continuously differential functions on $\mathbb{R}$ with bounded derivatives. Consider the solution $X=\left\{X_{t}, t \in [0,T]\right\}$ of the stochastic differential equation
\begin{eqnarray}
\nonumber X(t) = x(0) + \int^{t}_{0} \sigma(X(s)) dW(s) + \int^{t}_{0} b(X(s))ds.
\end{eqnarray}
Then, we have
\begin{eqnarray}
\nonumber D_{s}X(t) = \sigma(X(s))\exp\left(\int^{t}_{s} \sigma'(X(s)) dW(s) + \int^{t}_{s} \lambda(s) ds\right){1\!\!1}_{[0,t]}(s).
\end{eqnarray}
where $\lambda(s)=[b' - \frac{1}{2}\left(\sigma'\right)^{2}](X(s))$.
\end{proposition}
\begin{proof}
See \cite{N}, Section 2.2.
\end{proof}

\section{Decomposition formula using Malliavin calculus.}

\hspace{0.4cm}In this section, we use the Malliavin calculus to extend the call option price decomposition in an anticipative framework. This time, the decomposition formula has one term less than in the Itô formula's setup.\\

\par

We recall the definition of the future average volatility as
		\begin{eqnarray}
		\nonumber \bar{\sigma}(t):=\sqrt{\frac{1}{T-t} \int^{T}_{t} \sigma^{2}(s) ds}.
		\end{eqnarray}

\begin{theorem} \label{Teo-Malliavin}
(Decomposition formula) For all $t\in[0,T)$, we have 
\begin{eqnarray}
\nonumber	V(t)&=& \mathbb{E}_{t}\left[BS(t,S(t),\bar{\sigma}(t))\right]\\ \nonumber
&+& \frac{1}{2}\mathbb{E}_{t}\left[\int^{T}_{t} e^{-ru} G(u,S(u), \bar{\sigma}(u)) \left(L^{2}(u,S(u),\sigma(u))- \sigma^2(u)\right)du\right]\\ \nonumber
&+&\frac{\rho}{2}\mathbb{E}_{t}\left[\int^{T}_{t} e^{-r(u-t)}L(u,S(u),\sigma(u)) H(u,S(u), \bar{\sigma}_u) \left(\int^{T}_{u}D^{W}_{u} \sigma^{2}(r) dr\right)du\right].
\end{eqnarray}
where

$$G(t, S(t), \sigma(t)):= S^{2}(t) \partial^{2}_{S} BS (t, S(t), \sigma(t)),$$

$$H(t, S(t), \sigma(t)):= S(t)\partial_{S}G(t, S(t), \sigma(t)),$$ 
and 

$$L(t,S(t), \sigma(t)):=\frac{\theta(t,S(t),\sigma(t))}{S(t)} .$$
\end{theorem}

\begin{proof}
Notice that $e^{-rT} BS(T,S(T),\bar{\sigma}(T))=e^{-rT} V_{T}$. As $e^{-rt}V(t)$ is a martingale we can write
\begin{eqnarray}
\nonumber e^{-rt}V(t) = \mathbb{E}_{t}\left(e^{-rT}V(T)\right)= \mathbb{E}_{t}\left(e^{-rT}BS(T, S(T), \bar{\sigma}(T))\right).
\end{eqnarray}
So, using the aproximation to the identity argument and then applying the Itô formula presented in Proposition \ref{MallaivinIto} to $$e^{-rt}BS(t,S(t), \bar{\sigma}(t)).$$ We deduce and using \eqref{rel-vega} and \eqref{FK} that 
\begin{eqnarray}	
\nonumber && e^{-rT}BS(T,S(T),\bar{\sigma}(T))  - e^{-rt}BS(t,S(t),\bar{\sigma}(t))=\\ \nonumber
&=& \int^{T}_{t} e^{-ru} \mathcal{L}_{\bar{\sigma}S}BS(u,S(u),\bar{\sigma}(u)) du \\ \nonumber
&+& \frac{1}{2}\int^{T}_{t} e^{-ru} S^{2}(u) \partial^{2}_{S}BS(u,S(u),\bar{\sigma}(u)) \left(\frac{\theta(u,S(u),\sigma(u))}{S(u)}\right)^{2}du\\ \nonumber
&-& \frac{1}{2}\int^{T}_{t} e^{-ru} S^{2}(u) \partial^{2}_{S}BS(u,S(u),\bar{\sigma}(u))  \sigma^{2}(u)du\\ \nonumber
&+&\int^{T}_{t} e^{-ru} \partial_{S}BS(u,S(u),\bar{\sigma}_{u}) \theta(u,S(u), \sigma(u)) \left(\rho dW(u)  + \sqrt{1-\rho^{2}}dB(u) \right)\\ \nonumber
&+&\frac{\rho}{2}\int^{T}_{t} e^{-ru} \theta(u,S(u), \sigma(u))\partial_{S}\left(S^{2}(u)\partial^{2}_{S}BS(u,S(u),\bar{\sigma}(u))\right) \left(\int^{T}_{u}D^{W}_{u} \sigma^{2}(r) dr\right) du.	
\end{eqnarray}

Taking conditional expectation and multiplying by $e^{rt}$, we have
\begin{eqnarray}	
\nonumber && \mathbb{E}_{t}[e^{-r(T-t)}BS(T,S(T),\bar{\sigma}(T))] = \mathbb{E}_{t}\left[BS(t,S(t),\bar{\sigma}(t))\right]\\ \nonumber
&+& \frac{1}{2}\mathbb{E}_{t}\left[\int^{T}_{t} e^{-r(u-t)}  G(u,S(u), \bar{\sigma}(u)) \left(L^{2}(u,S(u),\sigma(u))- \sigma^{2}(u)\right)du\right]\\ \nonumber
&+&\frac{\rho}{2}\mathbb{E}_{t}\left[\int^{T}_{t} e^{-r(u-t)}L(u,S(u),\sigma(u)) H(u,S(u), \bar{\sigma}(u)) \left(\int^{T}_{u}D^{W}_{u} \sigma^{2}(r)dr\right) du\right].	
\end{eqnarray}
\end{proof}

\begin{Rem}
As it is expected, a new term emerges when it is considered \eqref{general_model} like it happen in Theorem \ref{Teo-Ito}.
\end{Rem}

\begin{Rem}
In particular, when $\theta(t,S(T), \sigma(t)) = \sigma(t) S(t)$
\begin{eqnarray}
\nonumber	V(t)&=& BS(t,S_{t},\bar{\sigma}(t))\\ \nonumber
&+&\frac{\rho}{2}\mathbb{E}_{t}\left[\int^{T}_{t} e^{-r(u-t)} \sigma(u)H(u,S(u), \bar{\sigma}_u)\left(\int^{T}_{u}D^{W}_{u} \sigma^{2}(r) dr\right) du\right].
\end{eqnarray}
Also, the gamma effect is cancelled as we have seen in the Itô formula section.
\end{Rem}

\begin{Rem}
Note that when $v(t)$ is a deterministic function, we have that all decomposition formulas are equal.
\end{Rem}

\begin{Rem}\label{Correl}
When $\rho=0$, we have
\begin{eqnarray}
\nonumber	&& \mathbb{E}_{t}\left[BS(t,S(t),\bar{\sigma}(t)) - BS(t,S_{t},v(t))\right]\\ \nonumber
&=&\frac{1}{2}\mathbb{E}_{t}\left[\int^{T}_{t} e^{-r(u-t)} \left(G(u,S(u), \bar{\sigma}(u)) - G(u, S(u), v(u))\right) L^{2}(u,S(u),\sigma(u))du\right]\\ \nonumber
&-&\frac{1}{2}\mathbb{E}_{t}\left[\int^{T}_{t} e^{-r(u-t)} \left(G(u,S(u), \bar{\sigma}(u)) - G(u, S(u), v(u))\right) \sigma^{2}(u)du\right]\\ \nonumber
&-&\frac{1}{8}\mathbb{E}_{t}\left[\int^{T}_{t} e^{-r(u-t)} K(u, S(u), v(u)) d\left[M,M\right](u)\right].
\end{eqnarray}
In particular, when $\theta(t,S(t), \sigma(t)) = \sigma(t)S(t)$ :
\begin{eqnarray*}
\nonumber	&& \mathbb{E}_{t}\left[BS(t,S(t),\bar{\sigma}(t)) - BS(t,S_{t},v(t))\right]\\
&=&-\frac{1}{8}\mathbb{E}_{t}\left[\int^{T}_{t} e^{-r(u-t)} K(u, S(u), v(u)) d\left[M,M\right](u)\right].
\end{eqnarray*}
The difference between the two appoaches is given by the vol vol of the option.
\end{Rem}

\section{An expression for the derivative of the implied volatility.}

\hspace{0.4cm}In this section, we give a general expression for the derivative of the implied volatility under the framework of Itô calculus and Malliavin calculus. There exist a previous calculation of this derivative in the case of exponential models using Malliavin calculus in \cite{ALV}.\\

\par

Let $I(S(t))$ denote the implied volatility process, which satisfies by definition $V(t)=BS(t, S(t), I(S(t)))$. We calculate the derivative of the implied volatility in the standard Itô case.
\begin{proposition} 
Under \eqref{general_model}, for every fixed $t\in[0,T)$ and assuming that  $(v(t))^{-1} < \infty$ a.s., we have that
\begin{eqnarray}
\nonumber \partial_{S}I(S^{*}(t)) &=& \frac{\mathbb{E}_{t}\left[\int^{T}_{t}\partial_{S} F_{2}(u, S^{*}(u), v(u))du\right]}{\partial_{\sigma} BS(t, S^{*}(t), I(S^{*}(t)))} \\ \nonumber
  &-& \frac{\mathbb{E}_{t}\left[\int^{T}_{t}\left(F_{1}(u, S^{*}(u), v(u)) + \partial_{S} F_{3}(u, S^{*}(u), v(u))\right) du\right]}{2S\partial_{\sigma} BS(t, S^{*}(t), I(S^{*}(t)))}.
\end{eqnarray}
where
\begin{eqnarray}
\nonumber && \mathbb{E}_{t}\left[\int^{T}_{t}F_{1}(u, S(u), v(u))du\right]\\ \nonumber
&=& \frac{1}{2}\mathbb{E}_{t}\left[\int^{T}_{t} e^{-r(u-t)} G(u, S(u), v(u)) \left(L^{2}(u,S(u),\sigma(u))- \sigma^{2}(u)\right)du\right]\\ \nonumber
&+&\frac{1}{8}\mathbb{E}_{t}\left[\int^{T}_{t} e^{-r(u-t)} K(u, S(u), v(u)) d\left[M,M\right](u)\right]\\ \nonumber
&+&\frac{\rho}{2}\mathbb{E}_{t}\left[\int^{T}_{t} e^{-r(u-t)} \frac{\theta(u,S(u), \sigma(u))}{S(u)} H(u, S(u), v(u)) d\left[W,M\right](u)\right],
\end{eqnarray}
\begin{eqnarray}
\nonumber &&\mathbb{E}_{t}\left[\int^{T}_{t}F_{2}(u, S(u), v(u))du\right]\\ \nonumber &=& \frac{\rho}{2}\mathbb{E}_{t}\left[\int^{T}_{t} e^{-r(u-t)} \frac{\theta(u,S(u), \sigma(u))}{S(u)} H(u, S(u), v(u)) d\left[W,M\right](u)\right]
\end{eqnarray}
and
\begin{eqnarray}
\nonumber &&\mathbb{E}_{t}\left[\int^{T}_{t}F_{3}(u, S(u), v(u))du\right]\\ 
\nonumber&=& \frac{1}{2}\mathbb{E}_{t}\left[\int^{T}_{t} e^{-r(u-t)} G(u, S(u), v(u)) \left(L^{2}(u,S(u),\sigma(u))- \sigma^{2}(u)\right)du\right]\\ \nonumber
&+&\frac{1}{8}\mathbb{E}_{t}\left[\int^{T}_{t} e^{-r(u-t)} K(u, S(u), v(u)) d\left[M,M\right](u)\right].
\end{eqnarray}
\end{proposition} 
\begin{proof}
Taking partial derivatives with respect to $S(t)$ on the expression $V(t)=BS(t, S(t), I(S(T)))$, we obtain
\begin{eqnarray}\label{implied1}
\partial_{S}V(t)= \partial_{S} BS(t, S(t), I(S(T))) + \partial_{\sigma} BS(t, S(t), I(S(T))) \partial_{S} I(S(t)).
\end{eqnarray}
On the other hand, from Theorem \ref{Teo-Ito} we deduce that
\begin{eqnarray}
V(t)= BS(t,S(t),v(t)) + \mathbb{E}_{t}\left[\int^{T}_{t}F_{1}(u, S(u), v(u))du\right],
\end{eqnarray}
which implies that
\begin{eqnarray}\label{implied2}
\partial_{S} V(t)= \partial_{S}BS(t,S(t),v(t)) + \mathbb{E}_{t}\left[\int^{T}_{t}\partial_{S} F_{1}(u, S(u), v(u))du\right].
\end{eqnarray}
Using that $(v(t))^{-1} < \infty$ we can check that $\partial_{S} V(t)$ is well define and finite a.s. Thus, using that $S^{*}(t)=K\exp(r(T-t))$, (\ref{implied1}) and (\ref{implied2}), we obtain
\begin{eqnarray}
\nonumber \partial_{S} I(S^{*}(t)) &=& \frac{\partial_{S}BS(t,S^{*}(t),v(t)) -\partial_{S} BS(t, S^{*}(t), I(S(t))) }{\partial_{\sigma} BS(t, S^{*}(t), I(S(t)))}\\ \nonumber
&+& \frac{\mathbb{E}_{t}\left[\int^{T}_{t}\partial_{S} F_{1}(u, S^{*}(u), v(u))du\right]}{\partial_{\sigma} BS(t, S^{*}(t), I(S(t)))}.
\end{eqnarray}
From \cite{RT} we know that $\partial_{S} I^{0}(t) =0$, where  $I^{0}(t)$ is the implied volatility in the case $\rho=0$, so
\begin{eqnarray}
\nonumber \partial_{S}BS(t,S^{*}(t),v(t)) &=& \partial_{S} BS(t, S^{*}(t), I^{0}(S(t)))-\mathbb{E}_{t}\left[\int^{T}_{t}\partial_{S} F_{3}(u, S^{*}(u), v(u))du\right].
\end{eqnarray}
So, we have that
\begin{eqnarray}
\nonumber \partial_{S} I(S^{*}(t)) &=& \frac{\partial_{S}BS(t,S^{*}(t),I^{0}(t)) -\partial_{S} BS(t, S^{*}(t), I(S^{*}(t))) }{\partial_{\sigma} BS(t, S^{*}(t), I(S^{*}(t)))}\\ \nonumber
&+& \frac{\mathbb{E}_{t}\left[\int^{T}_{t}\partial_{S} F_{2}(u, S^{*}(u), v(u))du\right]}{\partial_{\sigma} BS(t, S^{*}(t), I(S^{*}(t)))}.
\end{eqnarray}
On the other hand, we have that
\begin{eqnarray}
\nonumber \partial_{S} BS(t, S^{*}(t), v(t)) = \phi(d)
\end{eqnarray}
and 
\begin{eqnarray}
\nonumber BS(t, S^{*}(t), v(t)) = S\left(\phi(d) - \phi(-d)\right)
\end{eqnarray}
where $\phi$ is the standard Gaussian density.
Then
\begin{eqnarray}
\nonumber \partial_{S}BS(t, S^{*}(t), v(t)) = \frac{BS(t, S^{*}(t), v(t)) + S}{2S}
\end{eqnarray}
and
\begin{eqnarray}
\nonumber && \partial_{S}BS(t,S^{*}(t),I^{0}(t)) -\partial_{S} BS(t, S^{*}(t), I(S^{*}(t)))\\ \nonumber
&=& \frac{1}{2S}\left(BS(t,S^{*}(t),I^{0}(t)) -BS(t, S^{*}(t), I(S^{*}(t)))\right)\\ \nonumber
&=& -\frac{1}{2S} \mathbb{E}_{t}\left[\int^{T}_{t}\left(F_{1}(u, S^{*}(u), v(u)) + \partial_{S} F_{3}(u, S^{*}(u), v(u))\right) du\right].
\end{eqnarray}
\end{proof}

Now, we derive the implied volatility using the Malliavin calculus. It has been proved in \cite{ALV}, in the case when $\theta(t,S(T), \sigma(t)) = \sigma(t) S(t)$.

\begin{proposition} 
Under \eqref{general_model}, for every fixed $t\in[0,T)$, assuming that $(\tilde{\sigma}(t))^{-1} < \infty$ a.s. Then we have that
\begin{eqnarray}
\nonumber \partial_{S}I(S^{*}(t)) &=& \frac{\mathbb{E}_{t}\left[\int^{T}_{t}\partial_{S} F_{2}(u, S^{*}(u), \bar{\sigma}(u))du\right] }{\partial_{\sigma} BS(t, S^{*}(t), I(S^{*}(t)))}\\ \nonumber
 &-&\frac{\mathbb{E}_{t}\left[\int^{T}_{t}\left(F_{1}(u, S^{*}(u), \bar{\sigma}(u)) + \partial_{S} F_{3}(u, S^{*}(u),\bar{\sigma}(u)) \right)du\right]}{2S\partial_{\sigma} BS(t, S^{*}(t), I(S^{*}(t)))}.
\end{eqnarray}
where
\begin{eqnarray}
\nonumber &&\mathbb{E}_{t}\left[\int^{T}_{t}F_{1}(u, S(u), \bar{\sigma}(u))du\right]\\ \nonumber
&=& \frac{1}{2}\mathbb{E}_{t}\left[\int^{T}_{t} e^{-r(u-t)} G(u, S(u), \bar{\sigma}(u)) \left(L^{2}(u,S(u),\sigma(u))- \sigma^{2}(u)\right)du\right]\\ \nonumber
&+&\frac{\rho}{2}\mathbb{E}_{t}\left[\int^{T}_{t} e^{-r(u-t)} L(u,S(u),\sigma(u)) H(u, S(u), \bar{\sigma}(u)) \left(\int^{T}_{u}D^{W}_{u} \sigma^{2}(r) dr\right)d\left[W,M\right](u)\right],
\end{eqnarray}
\begin{eqnarray}
\nonumber &&\mathbb{E}_{t}\left[\int^{T}_{t}F_{2}(u, S(u), \bar{\sigma}(u))du\right]\\ \nonumber &=& \frac{\rho}{2}\mathbb{E}_{t}\left[\int^{T}_{t} e^{-r(u-t)}L^{2}(u,S(u),\sigma(u)) \left(\int^{T}_{u}D^{W}_{u} \sigma^{2}(r) dr\right)d\left[W,M\right](u)\right]
\end{eqnarray}
and
\begin{eqnarray}
\nonumber &&\mathbb{E}_{t}\left[\int^{T}_{t}F_{3}(u, S(u), \bar{\sigma}(u))du\right]\\ \nonumber
&=& \frac{1}{2}\mathbb{E}_{t}\left[\int^{T}_{t} e^{-r(u-t)} G(u, S(u), \bar{\sigma}(u)) \left(L^{2}(u,S(u),\sigma(u))- \sigma^{2}(u)\right)du\right].
\end{eqnarray}  
\end{proposition} 

\begin{proof}
See \cite{ALV} or the previous proof.
\end{proof}

\begin{Rem}
Note that this is generalisation of the formula proved in \cite{ALV}. In that case, $F_{1}=F_{2}$ and $F_{3}=0$. 
\end{Rem}

\section{Examples}
\hspace{0.4cm} Now, we give some applications of the decomposition formula for well-known models in Finance.\par

\subsection{Heston Model.}
\hspace{0.4cm}We consider that the stock prices  follows the Heston Model \eqref{general_model}. Using Theorem \ref{Teo-Ito} or Theorem \ref{Teo-functional}, we have

\begin{eqnarray}
\nonumber	V(t) &=&  BS(t,X(t),v(t)) \\ \nonumber
&+& \frac{\rho}{2}\mathbb{E}_{t}\left[\int^{T}_{t}e^{-r(u-t)} H(u,X(u),v(u)) \left(\int^{T}_{u} e^{-k(r-s)} dr\right) \sigma^{2}(u)\nu du\right]\\ \nonumber
&+& \frac{1}{8} \mathbb{E}_{t}\left[\int^{T}_{t}e^{-r(u-t)}  K(u,X(u),v(u))\left(\int^{T}_{u} e^{-k(r-s)} dr\right)^{2} \nu^{2} \sigma^{2}(u) du\right].
\end{eqnarray}

Using Theorem \ref{Teo-Malliavin}, we have that 
\begin{eqnarray}
\nonumber	V(t)&=& BS(t,S_{t},\bar{\sigma}(t))\\ \nonumber
&+&\frac{\rho}{2}\mathbb{E}_{t}\left[\int^{T}_{t} e^{-r(u-t)} H(u,S(u), \bar{\sigma}(u))\left(\int^{T}_{u}D^{W}_{u} \sigma^{2}(r) dr\right)\sigma(u)du\right].
\end{eqnarray}
where $D^{W}_{u} \sigma^{2}(r) = \nu \sigma(u) \exp\left(\frac{\nu}{2}\int^{r}_{u} \frac{1}{\sigma(s)} dW(s) + \int^{r}_{u}\left[-k - \frac{\nu^{2}}{8\sigma^{2}(s)}\right]ds\right).$

\subsection{SABR Model}
\hspace{0.4cm} We consider that the stock prices follows the SABR model \eqref{SABR}. Using Theorem \ref{Teo-Ito} or Theorem \ref{Teo-functional}, we have

\begin{eqnarray}
\nonumber	V(t)&=& BS(t,S(t),v(t))\\ \nonumber
&+&\frac{1}{2}\mathbb{E}_{t}\left[\int^{T}_{t} e^{-r(u-t)} G(u, S(u), v(u)) \sigma^{2}(u)\left(S^{2(\beta-1)}(u)-1\right)du\right]\\ \nonumber
&+&\frac{1}{8}\mathbb{E}_{t}\left[\int^{T}_{t} e^{-r(u-t)} K(u, S(u), v(u)) d\left[M,M\right]\right](u)\\ \nonumber
&+&\frac{\rho}{2}\mathbb{E}_{t}\left[\int^{T}_{t} e^{-r(u-t)} H(u, S(u), v(u))\sigma(u)d\left[W,M\right](u)\right].
\end{eqnarray}
where 

$$d\left[M,M\right]= 4 \alpha^{2} \sigma^{4}(t) \left(\int^{T}_{t} e^{\alpha^{2}(s-t)} ds\right)^{2}dt$$ 
and 

$$d\left[M,W\right]= 2 \alpha \sigma^2 (t) \left(\int^{T}_{t} e^{\alpha^{2}(s-t)} ds\right)dt.$$

Using Theorem \ref{Teo-Malliavin}, we have that 
\begin{eqnarray}
\nonumber	V(t)&=& \mathbb{E}_{t}\left[BS(t,S(t),\bar{\sigma}(t))\right]\\ \nonumber
&+& \frac{1}{2}\mathbb{E}_{t}\left[\int^{T}_{t} e^{-r(u-t)} G(u,S(u), \bar{\sigma}(u))\sigma^{2}(u) \left(S^{2(\beta-1)}(u)-1\right)du\right]\\ \nonumber
&+&\frac{\rho}{2}\mathbb{E}_{t}\left[\int^{T}_{t} e^{-r(u-t)} H(u,S(u), \bar{\sigma}(u)) \left(\int^{T}_{u}D^{W}_{u} \sigma^{2}(r) dr\right) \sigma(u)du\right].
\end{eqnarray}
where $D^{W}_{u} \sigma^{2}(r)=2 \alpha \sigma^{2}(u){1\!\!1}_{[0,r]}(u).$

\section{Conclusion}
\hspace{0.4cm} In this paper, we notice that the idea used in \cite{Alos} can be used for a generic Stochastic Differential Equation (SDE). There is no need to specify the volatility process, only existence and uniqueness of the solution of the SDE are needed, allowing much more flexibility in the decomposition formula. We see the effect on assuming that the stock price follows a exponential process and how a new term arises in a general framework. Also, we have computed the decomposition using three different method: Itô formula, functional Itô calculus and  Malliavin calculus. In the case of call options, the idea used in \cite{Alos} is equivalent to the use of the functional Itô formula developed in \cite{Cont1, Cont2, Cont3, Dupire}, but without the need of the theory behind the functional Itô calculus. Both formulas can be applied to Dirichlet process, in particular, to the fractional Brownian motion with Hurst parameter equal or bigger than $\frac{1}{2}$. Furthermore, we realize that the Feynman-Kac formula has a key role into the decompossition process.

{ }

\end{document}